\documentclass[11pt,a4paper]{article}
\usepackage{tikz,tkz-graph,amssymb,amsmath,amsthm,multirow,color,xcolor,cancel,bbm}
\usepackage[textheight=25cm,textwidth=16cm]{geometry}

\newtheorem{lemma}{Lemma}
\newtheorem{corollary}{Corollary}
\newtheorem{question}{Question}
\newtheorem{remark}{Remark}
\newtheorem{theorem}{Theorem}
\newtheorem{example}{Example}

\def\cst{\mathrm{cst}}
\def\B{\{0,1\}}
\def\In{\mathrm{In}}
\def\e{\epsilon}

\title{Fixed points and connections between positive and negative cycles in Boolean networks}

\author{
Adrien Richard\footnote{Laboratoire I3S, CNRS, Universit\'e C\^ote d'Azur \& Universit\'e Nice Sophia Antipolis, France.
\tt{richard@unice.fr}}
}

\date{November 28, 2016; Revised November 14, 2017}

\begin{document}

\maketitle

\begin{abstract} 
We are interested in the relationships between the number fixed points in a Boolean network $f:\B^ n\to\B^n$ and its interaction graph, which is the arc-signed digraph $G$ on $\{1,\dots,n\}$ that describes the positive and negative influences between the components of the network. A fundamental theorem of Aracena says that if $G$ has no positive (resp. negative) cycle, then $f$ has at most (resp. at least) one fixed point; the sign of a cycle being the product of the signs of its arcs. In this note, we generalize this result by taking into account the influence of connections between positive and negative cycles. In particular, we prove that if every positive (resp. negative) cycle of $G$ has an arc $a$ such that $G\setminus a$ has a non-trivial initial strongly connected component containing the terminal vertex of $a$ and only negative (resp. positive) cycles, then $f$ has at most (resp. at least) one fixed point. This is, up to our knowledge, the first generalization of Aracena's theorem where the conditions are expressed with $G$ only.

%Besides, Aracena proved that if $G$ is strongly connected and has no negative cycle, then $f$ has two fixed points with Hamming distance $n$, and we prove that the same conclusion can be obtained under the following condition: $G$ is strongly connected, has a unique negative cycle $C$, has at least one positive cycle, and $f$ canalizes no arc of $C$. 

\end{abstract}

%%%%%%%%%%%%%%%%%%%%%%%%%%%%%%%%%%%%%%%%%%%%%%%%%%%%%%
\section{Introduction}
%%%%%%%%%%%%%%%%%%%%%%%%%%%%%%%%%%%%%%%%%%%%%%%%%%%%%%

A {\em Boolean network} with $n$ components is a discrete dynamical system usually defined by a global transition function
\[
f:\B^n\to\B^n,\qquad x=(x_1,\dots,x_n)\mapsto f(x)=(f_1(x),\dots,f_n(x)).
\]
Boolean networks have many applications. In particular, since the seminal papers of McCulloch and Pitts \cite{MP43}, Hopfield \cite{H82}, Kauffman \cite{K69,K93} and Thomas \cite{T73,TA90}, they are omnipresent in the modeling of neural and gene networks (see \cite{B08,N15} for reviews). They are also essential tools in information theory, for the network coding problem \cite{ANLY00,GRF16}. 

\medskip
The structure of a Boolean network $f$ is usually represented via its {\em interaction graph}, defined below using the following notion of derivative. For every $u,v\in\{1,\dots,n\}$, the {\em discrete derivative} of $f_v$ with respect to the variable $x_u$ is the function $f_{vu}:\B^n\to\{-1,0,1\}$ defined by 
\[
f_{vu}(x):=f_v(x_1,\dots,x_{u-1},1,x_{u+1},\dots,x_n)-f_v(x_1,\dots,x_{u-1},0,x_{u+1},\dots,x_n).
\]
The {\bf interaction graph} of $f$ is the signed digraph $G$ defined as follows: the vertex set is $[n]:=\{1,\dots,n\}$ and, for all $u,v\in[n]$, there is a positive (resp. negative) arc from $u$ to $v$ if $f_{vu}(x)$ is positive (resp. negative) for at least one $x\in\B^n$. Note that $G$ can have both a positive and a negative arc from one vertex to another (in that case, the sign of the interaction depends on the state $x$ of the system). In the following, an arc from $u$ to $v$ of sign $\e\in\{+,-\}$ is denoted $(uv,\e)$. Also, cycles are always directed and regarded as subgraphs (no repetition of vertices is allowed). The sign of a cycle is, as usual, defined as the product of the signs of its arcs.
 
\medskip 
In many contexts, as in molecular biology, the interaction graph is known $-$ or at least well
approximated $-$, while the actual dynamics described by $f$ is not, and is very difficult to observe \cite{TK01,N15}. A natural question is then the following. 

\begin{question}\label{question} 
What can be said on the dynamics of $f$ according to $G$ only? 
\end{question}

Among the many dynamical properties that can be studied, fixed points are of special interest, since they correspond to stable states and often have a strong meaning. For instance, in the context of gene networks, they correspond to stable patterns of gene expression at the basis of particular biological processes \cite{TA90,A04}. As such, they are arguably the property which has been the most thoroughly studied (see \cite{R86} for an introduction to these studies). In particular, many works studied sufficient conditions for the uniqueness or the existence of a fixed point \cite{SD05,A08,RRT08,RC07,R10,R15} (such questions have also been widely studied in the continuous setting, see \cite{KST07} and the references therein). 

\medskip
Here, we are mainly interested in the following two fundamental theorems, suggested by the biologist Thomas \cite{T81}, and known as the Boolean versions of the first and second Thomas' rules.

\begin{theorem}\label{rule1}
If $G$ has no positive cycle, then $f$ has at most one fixed point. More generally, if $f$ has two distinct fixed points $x$ and $y$, then $G$ has a positive cycle $C$ such that $x_v\neq y_v$ for every vertex $v$ in $C$.
\end{theorem}

\begin{theorem}\label{rule2}
If $G$ has no negative cycle, then $f$ has at least one fixed point.
\end{theorem}

The first theorem has been proved by Aracena, see the proof of \cite[Theorem 9]{A08} and also \cite{ADG04a}, and a stronger version has been independently proved by Remy, Ruet and Thieffry \cite[Theorem 3.2]{RRT08}. The second theorem is an easy application of another result of Aracena \cite[Theorem 6]{A08}.% and has been proved independently in \cite[Corollary 1]{R10}.

\medskip
Two upper bounds on the number of fixed points can be deduced from Theorem~\ref{rule1}. Let~$\tau^+$ be the minimal number of vertices whose deletion in $G$ leaves a signed digraph without positive cycle. From Theorem~\ref{rule1} we deduce, using arguments reproduced below in the proof of Corollary~\ref{newupperbound}, that {\em $f$ has at most $2^{\tau^+}$ fixed points} \cite[Theorem 9]{A08}. For the second bound, two additional definitions are needed. Let $g^+$ be the minimum length of a positive cycle of $G$ (with the convention that $g^+=\infty$ if $G$ has no positive cycle), and for every integer $d$, let $A(n,d)$ be the maximal size of a subset $X\subseteq\B^n$ such that the Hamming distance between any two distinct elements of $X$ is at least $d$. According to Theorem~\ref{rule1}, the Hamming distance between any two distinct fixed points of $f$ is at least $g^+$, and thus we get a second upper bound: {\em $f$ has at most $A(n,g^+)$ fixed points}. The quantity $A(n,d)$, usually called {\em maximal size of a binary code of length $n$ with minimal distance $d$}, has been intensively studied in Coding Theory. The well known Gilbert bound and sphere packing bound give the following approximation: $2^{n}/\sum_{k=0}^{d-1}{n \choose k}\leq A(n,d)\leq 2^{n}/\sum_{k=0}^D{n \choose k}$ with $D={\lfloor\frac{d-1}{2}\rfloor}$. See \cite{GRR15} for other connections with Coding Theory. 

%\cite{R86,ADG04,SD05,RRT08,A08,R09,ARS14,GRR14,R15}

\medskip
All the generalizations of the previous results known so far use additional information on $f$ \cite{RRT08,R10} (or consist in enlarging the framework, considering discrete networks instead of Boolean networks and asynchronous attractors instead of fixed points \cite{RC07,R10}). In this note, we establish, up to our knowledge, the first generalizations that only use information on $G$, and which thus contribute directly to Question~\ref{question} (these are the Theorems~\ref{newrule1}, \ref{renewrule1} and \ref{newrule2} stated below).

\medskip
Our approach is the following. The previous results show that positive and negative cycles are key structures to understand the relationships between $G$ and the fixed points of $f$. However, they use information on positive cycles only, or on negative cycles only. It is then natural to think that improvements could be obtained by considering the two kinds of cycles simultaneously. This is what we do here, by highlighting two qualitative phenomena on the influence of connections between positive and negative cycles. These two dual phenomena could be verbally described as follows (we say that two graphs {\em intersect} if they share a common vertex): 
\begin{enumerate}
\item
{\em If each positive cycle $C$ of $G$ intersects a negative cycle $C'$, and if $C$  ``isolates'' $C'$ from the other positive cycles, then $f$ behaves as in the absence of positive cycles: it has at most  one fixed point.} 
\item
{\em If each negative cycle $C$ of $G$ intersects a positive cycle $C'$, and if $C$  ``isolates'' $C'$ from the other negative cycles, then $f$ behaves as in the absence of negative cycles: it has at least one fixed point.}
\end{enumerate}

The following three theorems give a support to these phenomena. Theorems~\ref{newrule1} and \ref{renewrule1} are uniqueness results that generalize (the first assertion in) Theorem~\ref{rule1}. Theorem~\ref{newrule2} is an existence result that generalizes Theorem~\ref{rule2} and shows, together with Theorem~\ref{newrule1}, an explicit duality between positive and negative cycles (all the notions involved in these statements are formally defined in the next section). 

\begin{theorem}\label{newrule1}
If every positive cycle of $G$ has an arc $a=(uv,\epsilon)$ such that $G\setminus a$ has a non-trivial initial strong component containing $v$ and only negative cycles, then $f$ has at most one~fixed~point.
\end{theorem}

\begin{theorem}\label{renewrule1}
If every positive cycle $C$ of $G$ has a vertex $v$ of in-degree at least two that belongs to no other positive cycle and with only in-neighbors in $C$, then $f$ has at most one fixed~point. 
\end{theorem} 

\begin{theorem}\label{newrule2}
If every negative cycle of $G$ has an arc $a=(uv,\epsilon)$ such that $G\setminus a$ has a non-trivial initial strong component containing $v$ and only positive cycles, then $f$ has at least one fixed point.
\end{theorem}

\begin{example}\label{ex1}
{\em Suppose that $G$ is the signed digraph described below. Each positive cycle $C$ is of length three and contains a vertex $v\in\{3,5,\dots,n\}$ with a negative loop. The arc $(v-1)\to v$ of $C$ then satisfies the condition of Theorem \ref{newrule1}, and the vertex $v$ satisfies the condition of Theorem~\ref{renewrule1}. Thus $f$ has at most one fixed point.
\begin{center}
\begin{tikzpicture}
\node (1) at (0,1){1};
\node (2) at (2,1){2};
\node (3) at (1,0){3};
\node (4) at (4,1){4};
\node (5) at (3,0){5};
\node (6) at (6,1){6};
\node (7) at (5,0){7};
\node (...) at (7,1){$\cdots\quad$};
\node (n3) at (8,1){$n-3$};
\node (n1) at (10,1){$n-1$};
\node (n) at (9,0){$n$};
\draw[->,thick] (3.-60) .. controls (1.5,-0.7) and (0.5,-0.7) .. (3.-120);
\node at (1,-0.7){{\scriptsize$-$}};
\draw[->,thick] (5.-60) .. controls (3.5,-0.7) and (2.5,-0.7) .. (5.-120);
\node at (3,-0.7){{\scriptsize$-$}};
\draw[->,thick] (7.-60) .. controls (5.5,-0.7) and (4.5,-0.7) .. (7.-120);
\node at (5,-0.7){{\scriptsize$-$}};
\draw[->,thick] (n.-60) .. controls (9.5,-0.7) and (8.5,-0.7) .. (n.-120);
\node at (9,-0.7){{\scriptsize$-$}};
\path[->,thick,bend left=15]
(1) edge node[above=-0.5mm] {{\small $+$}} (2)
(2) edge node[near start,below=0.2mm] {{\scriptsize $+$}} (3)
(3) edge node[near end,below=0.2mm] {{\scriptsize $+$}} (1)
(2) edge node[above=-0.5mm] {{\scriptsize $+$}} (4)
(4) edge node[near start,below=0.2mm] {{\scriptsize $+$}} (5)
(5) edge node[near end,below=0.2mm] {{\scriptsize $+$}} (2)
(4) edge node[above=-0.5mm] {{\scriptsize $+$}} (6)
(6) edge node[near start,below=0.2mm] {{\scriptsize $+$}} (7)
(7) edge node[near end,below=0.2mm] {{\scriptsize $+$}} (4)
(n3) edge node[above=-0.5mm] {{\scriptsize $+$}} (n1)
(n1) edge node[near start,below=0.2mm] {{\scriptsize $+$}} (n)
(n)  edge node[near end,below=0.2mm] {{\scriptsize $+$}} (n3)
;
\end{tikzpicture}
\end{center}}
\end{example}

\begin{remark}\label{rem:upperbound}
{\em Theorems~\ref{newrule1} and \ref{renewrule1} are already relevant compared to the upper bounds $2^{\tau^+}$ and $A(n,g^+)$. Indeed, if $G$ is signed digraph described above,  then $\tau^+=\lceil\frac{n-1}{4}\rceil$, thus the first upper-bound is exponential with $n$, and $g^+=3$ so the second upper bound is at least $\frac{2^{n}}{n^2+1}$ by the Gilbert bound, and is also exponential with $n$. However, as explained above, $G$ satisfies the conditions of Theorems~\ref{newrule1} and \ref{renewrule1}, which give an upper bound equal to one only.}
\end{remark}

\begin{example}\label{ex2}
{\em Suppose that $G$ is the signed digraph described below. Each negative cycle $C$ is of length three and contains a vertex $v\in\{3,5,\dots,n\}$ with a positive loop. The arc $(v-1)\to v$ of $C$ then satisfies the condition of Theorem \ref{newrule2}. Thus $f$ has at least one fixed point.
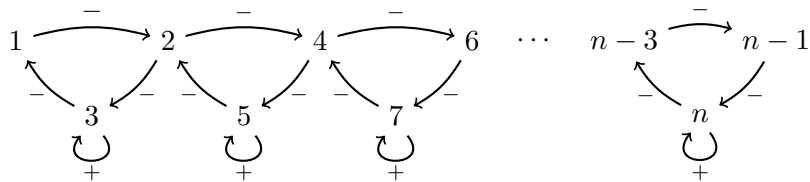
\begin{center}
\begin{tikzpicture}
\node (1) at (0,1){1};
\node (2) at (2,1){2};
\node (3) at (1,0){3};
\node (4) at (4,1){4};
\node (5) at (3,0){5};
\node (6) at (6,1){6};
\node (7) at (5,0){7};
\node (...) at (7,1){$\cdots\quad$};
\node (n3) at (8,1){$n-3$};
\node (n1) at (10,1){$n-1$};
\node (n) at (9,0){$n$};
\draw[->,thick] (3.-60) .. controls (1.5,-0.7) and (0.5,-0.7) .. (3.-120);
\node at (1,-0.75){{\scriptsize$+$}};
\draw[->,thick] (5.-60) .. controls (3.5,-0.7) and (2.5,-0.7) .. (5.-120);
\node at (3,-0.75){{\scriptsize$+$}};
\draw[->,thick] (7.-60) .. controls (5.5,-0.7) and (4.5,-0.7) .. (7.-120);
\node at (5,-0.75){{\scriptsize$+$}};
\draw[->,thick] (n.-60) .. controls (9.5,-0.7) and (8.5,-0.7) .. (n.-120);
\node at (9,-0.75){{\scriptsize$+$}};
\path[->,thick,bend left=15]
(1) edge node[above=-0.5mm] {{\small $-$}} (2)
(2) edge node[near start,below=0.2mm] {{\scriptsize $-$}} (3)
(3) edge node[near end,below=0.2mm] {{\scriptsize $-$}} (1)
(2) edge node[above=-0.5mm] {{\scriptsize $-$}} (4)
(4) edge node[near start,below=0.2mm] {{\scriptsize $-$}} (5)
(5) edge node[near end,below=0.2mm] {{\scriptsize $-$}} (2)
(4) edge node[above=-0.5mm] {{\scriptsize $-$}} (6)
(6) edge node[near start,below=0.2mm] {{\scriptsize $-$}} (7)
(7) edge node[near end,below=0.2mm] {{\scriptsize $-$}} (4)
(n3) edge node[above=-0.5mm] {{\scriptsize $-$}} (n1)
(n1) edge node[near start,below=0.2mm] {{\scriptsize $-$}} (n)
(n)  edge node[near end,below=0.2mm] {{\scriptsize $-$}} (n3)
;
\end{tikzpicture}
\end{center}}
\end{example}

From a technical point of view, proofs are mainly based on refinements of Aracena's arguments. The innovation principally relies in the nature of the statements, which generalize previous results by taking into account both kinds of cycles simultaneously, and only use information on $G$. For uniqueness results, we will prove a theorem that generalizes Theorem~\ref{rule1} and that directly implies Theorems~\ref{newrule1}, \ref{renewrule1} as well as upper bounds that only depend on $G$ and that improve the two upper-bounds discussed above. 

\medskip
Let us say that a positive (resp. negative) arc of $G$ from $u$ to $v$ is {\bf canalized} by $f$ if there exists $c\in\B$ such that $f_v(x)=c$ for all $x\in\B^n$ with $x_u=c$ (resp. $x_u\neq c$). The proofs reveal interesting properties on networks without canalized arc. In particular, using mainly graph theoretic arguments, we will prove the following theorem, established by Aracena under the hypothesis that $G$ is strongly connected and has no negative cycle \cite[Theorem 6]{A08}. 

\begin{theorem}\label{strong}
If $G$ is strongly connected, has a unique negative cycle and at least one positive cycle, and if $f$ canalizes no arc that belongs to the negative cycle, then $f$ has two fixed points with Hamming distance $n$. 
\end{theorem}

Let us mention the few other theoretical works we know that consider connections between positive and negative cycles in Boolean networks. They differ from the present work by the hypothesis made on $G$. In \cite{DNS12} and \cite{MNR14} is given a comprehensive analysis of the synchronous and asynchronous dynamics of Boolean networks whose interaction graph $G$ consists of two cycles that share exactly one vertex ($G$ is then a so-called double-cycle). Both in the synchronous and asynchronous case, the dynamics is easier to understand when both cycles are positive than when the two cycles have different signs, and the most intriguing and difficult case occurs when both cycles are negative. Besides, \cite{DR12} proposes a study of the number of fixed points under the following hypothesis: $G$ is strongly connected, has a vertex $v$ meeting every cycle, and all the vertices $u\neq v$ are of in-degree one (so $\tau^+= 1$ and $f$ has thus at most two fixed points). The results of \cite{DR12} not contained in \cite{A08} are then essentially the following. Firstly, if $G$ has a unique positive (resp. negative) cycle and at least one negative (resp. positive) cycle, then $f$ has at most (resp. at least) one fixed point (this is an easy consequence of Theorems~\ref{newrule1} and \ref{newrule2}, as explained below, in Remarks \ref{rem:DR1} and \ref{rem:DR2}). Secondly, if $G$ has at least two positive cycles and two negative cycles, then $f$ may have zero, one or two fixed points. Other works based on simulations and considering connections between positive and negative cycles can be found in \cite{KC07,SVL08}.

\medskip 
The paper is organized as follows. Preliminaries are given in Section~\ref{sec:preliminaries}. Uniqueness results and the resulting upper-bounds are given in Sections~\ref{sec:uniqueness} and \ref{sec:upperbound}. Existence results are given in Section~\ref{sec:existence}. Concluding remarks are given in Section~\ref{sec:conclusion}. 

%%%%%%%%%%%%%%%%%%%%%%%%%%%%%%%%%%%%%%%%%%%%%%%%%%%%%%
\section{Preliminaries}\label{sec:preliminaries}
%%%%%%%%%%%%%%%%%%%%%%%%%%%%%%%%%%%%%%%%%%%%%%%%%%%%%%

A signed digraph $G$ has a vertex set $V(G)$ and an arc set $A(G)$. Each arc $a\in A(G)$ has an initial vertex $v$, a terminal vertex $u$, a {\em sign} $\e\in\{+,-\}$, and is written as $a=(uv,\e)$. We also say that $a=(uv,\e)$ is an arc from $u$ to $v$, that $a$ is an {\em in-coming arc} of $v$ of sign $\e$, and that $u$ is an {\em in-neighbor} of $v$ of sign $\e$. The {\em in-degree} of $v$ is the number of in-coming arcs of $v$. A {\em source} is a vertex of in-degree zero. The set of in-neighbors of $v$ is denoted $\In(v)$, and the set of in-neighbors of sign $\e$ is denoted $\In^\e(v)$. We say that $G$ is {\em simple} if $\In^+(v)$ and $\In^-(v)$ are disjoint for every vertex $v$. We abusively write $G=\{v\}$ to mean that $G$ is the {\em trivial graph} with $v$ as unique vertex and no arc. If $G'$ is another signed digraph then $G\cup G'$ is the signed digraph with vertex set $V(G)\cup V(G')$ and arc set $A(G)\cup A(G')$. If $a\in V(G)^2\times \{+,-\}$ then $G\cup a$ has vertex set $V(G)$ and arc set $A(G)\cup\{a\}$. 

\medskip
A {\em subgraph} of $G$ is a signed digraph obtained from $G$ by removing arcs or vertices (with the attached arcs). We write $G'\subseteq G$ to mean that $G'$ is a subgraph of $G$. If $x$ is a vertex or an arc, then $G\setminus x$ is the subgraph obtained from $G$ by removing $x$. The subgraph of $G$ {\em induced} by a set of vertices $I\subseteq V(G)$, denoted $G[I]$, is the subgraph obtained from $G$ by removing every vertex not in $I$. We denote by $G^I$ the subgraph obtained from $G$ by removing every arc with a terminal vertex in $I$. 

\medskip
Paths and cycles of $G$ are always directed and regarded as simple subgraphs. The {\em sign} of a path or a cycle is the product of the signs of its arcs. Thus a path or a cycle is positive if and only if it contains an even number of negative arcs. A {\em strong component} (or component for short) of $G$ is a maximal set of vertices $I$ (with respect to the inclusion relation) such that $G[I]$ is strongly connected ({\em strong} for short). Such a component $I$ is {\em trivial} if $|I|=1$. If $G$ has no arc from $V(G)\setminus I$ to $I$ then $I$ is an {\em initial component}. If $G$ has no arc from $I$ to $V(G)\setminus I$, then $I$ is a {\em terminal component}. If $C$ is a cycle and $u,v\in V(C)$ then $C[u,v]$ is the path from $u$ to $v$ contained in $C$ (with the convention that this path is the trivial path $P=\{u\}$ if $u=v$). If $P$ is a path and if $u$ and $v$ are vertices in $P$ such that $v$ does not appear before $u$ in $P$, then $P[u,v]$ is the path from $u$ to $v$ contained in $P$.  

\medskip
In all the following, $f$ always denotes an $n$-component Boolean network, that is $f:\B^n\to\B^n$, and $G$ always denotes its interaction graph, as defined in the introduction. The vertex set of $G$ is thus $[n]:=\{1,\dots,n\}$. If $x\in\B^n$ and $I\subseteq [n]$ then $x_I$ is the restriction of $x$ to the components in $I$. The {\em Hamming distance} between two points $x,y\in\B^n$, denoted $d(x,y)$, is the number of $v\in [n]$ such that $x_v\neq y_v$. We denote by $\overline{x}^v$ the point $y\in\B^n$ that differs from $x$ only in $x_v\neq y_v$. We denote by $\overline{x}$ the point at Hamming distance $n$ from~$x$. 

\medskip
For all vertex $v$ of $G$, we define the partial order $\leq_v$ on $\B^n$ as follows: for all $x,y\in\B^n$, 
\[
x\leq_v y
\quad
\iff
\quad
x_{\In^+(v)}\leq y_{\In^+(v)}
\text{ and }
x_{\In^-(v)}\geq y_{\In^-(v)}.
\]
Note that if $x\leq_v y$ then $x_u=y_u$ for all $u\in\In^+(v)\cap\In^-(v)$. A very basic property is that $f_v$ is non-decreasing with respect to the partial order $\leq_v$. 

%%%%%%%%%%%%%
\begin{lemma}\label{local_lemma}
For all $v\in [n]$ and $x,y\in\B^n$,
\[
x\leq_v y
\quad
\Rightarrow
\quad
f_v(x)\leq f_v(y).
\]
\end{lemma}
%%%%%%%%%%%%%

\begin{proof}[{\bf Proof}]
Let $I(x,y)$ be the set of $u\in\In(v)$ such that $x_u\neq y_u$. We prove the lemma by induction on the size of $I(x,y)$. If $I(x,y)$ is empty, then we clearly have $f_v(x)=f_v(y)$ thus the lemma holds. For the induction step, suppose that $I(x,y)$ contains some vertex $u$ and that $x\leq_v y$. If $u\in\In^+(v)$, then $x_u<y_u$ thus $u\not\in\In^-(v)$. We deduce that $f_v(x)\leq f_v(\overline{x}^u)$, since otherwise $f_{vu}(x)<0$ and thus $u\in\In^-(v)$, a contradiction. Furthermore, $\overline{x}^u\leq_v y$, and since $I(\overline{x}^u,y)=I(x,y)\setminus\{u\}$, by induction hypothesis, $f_v(\overline{x}^u)\leq f_v(y)$. Hence, $f_v(x)\leq f_v(y)$ as required. If $u\in\In^-(v)$ the proof is similar.
\end{proof}

For all $x\in\B^n$, we denote by $G(x)$ the subgraph of $G$ with vertex set $[n]$ that contains all the  arcs $(uv,+)$ of $G$ with $x_u=x_v$ and all the arcs $(uv,-)$ of $G$ with $x_u\neq x_v$. Hence, a path of $G(x)$ from $u$ to $v$ is positive if $x_u=x_v$ and negative if $x_u\neq x_v$. As a consequence, we have the following basic property, used several times below.

\begin{lemma}
For all $x\in\B^n$, all the cycles of $G(x)$ are positive.
\end{lemma}

%%%%%%%%%%%%%%%%%%%%%%%%%%%%%%%%%%%%%%%%%%%%%%%%%%%%%%
\section{Uniqueness results}\label{sec:uniqueness}
%%%%%%%%%%%%%%%%%%%%%%%%%%%%%%%%%%%%%%%%%%%%%%%%%%%%%%

Uniqueness results are based on the following definition. An arc $a=(uv,\e)$ in a positive cycle $C$ is a {\bf special arc} of $C$ if the following holds in  $G\setminus a$:
\begin{itemize}
\item[$(i)$] $v$ is not a source,
\item[$(ii)$] $v$ is not contained in a positive cycle,
\item[$(iii)$] every path from a positive cycle or a source to $v$ intersects $C\setminus v$.
\end{itemize} 

The main result of this section is the following generalization of Theorem~\ref{rule1}. 

%%%%%%%%%%%%%
\begin{theorem}\label{uniqueness}
If $x$ and $y$ are distinct fixed points of $f$, then $G$ has a positive cycle $C$ with no special arc such that $x_v\neq y_v$ for every vertex $v$ in $C$. 
\end{theorem}
%%%%%%%%%%%%%

The following lemma has been proved in \cite[Theorem 9]{A08} and already implies Theorem~\ref{rule1}. We include a short proof for completeness.  

%%%%%%%%%%%%%%
\begin{lemma}\label{alphabeta_lemma}
Suppose that $x$ and $y$ are distinct fixed points of $f$. Then $G(x)$ has a cycle $C$ such that $x_v\neq y_v$ for all vertex $v$ in $C$. 
\end{lemma}
%%%%%%%%%%%%%%

\begin{proof}[{\bf Proof}]
Let $I$ be the set of $v\in[n]$ such that $x_v\neq y_v$. Suppose, for a contradiction, that $G(x)[I]$ has a source, say $v$. If $x_v=1$, it means that $x_u=0$ for all $u\in \In^+(v)\cap I$ and $x_u=1$ for all $u\in \In^-(v)\cap I$. Since $x_u=y_u$ for all $u\not\in I$, we deduce that $x\leq_v y$ and thus $f_v(x)\leq f_v(y)$. Since $f_v(x)=x_v=1$ we have $f_v(y)=1$ which is a contradiction since $x_v\neq y_v=f_v(y)$. If $x_v=0$ we obtain a contradiction similarly. Hence, $G(x)[I]$ has no source and thus contains at least one cycle. 
\end{proof}

%%%%%%%%%%%%%
\begin{lemma}\label{source_lemma}
If $f(x)=x$ then every source of $G(x)$ is a source of $G$. 
\end{lemma}
%%%%%%%%%%%%%

\begin{proof}[{\bf Proof}]
Suppose that $v$ is a source of $G(x)$. If $x_v=1$ it means that $x_u=0$ for all $u\in \In^+(v)$ and $x_u=1$ for all $u\in \In^-(v)$. Hence, for all $y\in\B^n$, we have $x\leq_v y$ and thus $f_v(x)\leq f_v(y)$. Since $f_v(x)=x_v=1$, we deduce that $f_v(y)=1$ for all $y\in\B^n$. Thus $f_v$ is constant and this is equivalent to say that $v$ is a source of $G$. If $x_v=0$ the proof is similar.  
\end{proof}

%%%%%%%%%%%%%
\begin{lemma}\label{canalizing}
Let $C$ be a positive cycle of $G$ with a special arc $a=(uv,\e)$. Let $x$ be a fixed point of $f$, and suppose that $C\subseteq G(x)$. Then $f_v(z)=x_v$ for all $z\in\B^n$ such that $z_u=x_u$. 
\end{lemma}
%%%%%%%%%%%%%

\begin{proof}[{\bf Proof}]
Let us prove that $v$ is a source of $G(x)\setminus a$. Suppose, for a contradiction, that this is not the case, and let $P$ be a path of $G(x)\setminus a$ from a source or a cycle of $G(x)\setminus a$ to $v$. By Lemma~\ref{source_lemma}, every source of $G(x)\setminus a$ is a source of $G\setminus a$, and since every cycle of $G(x)\setminus a$ is a positive cycle of $G\setminus a$, we deduce that $P$ is a path from a source or a positive cycle of $G\setminus a$ to $v$. We deduce from the definition of a special arc that $P$ intersects $C\setminus v$. Let $w$ be the last vertex of $P$ that belongs to $C\setminus v$. Then $C'=P[w,v]\cup C[v,w]$ is a cycle of $G(x)\setminus a$. Thus $C'$ is a positive cycle of $G\setminus a$ containing $v$, and this contradicts the fact that $a$ is a special arc. Thus $v$ is indeed a source of $G(x)\setminus a$. 

\medskip
Let $z\in\B^n$ with $z_u=x_u$. Suppose that $x_v=1$. Since $v$ is a source of $G(x)\setminus a$, we have $x_w=0$ for all $w\in\In^+(v)\setminus\{u\}$ and $x_w=1$ for all $w\in\In^-(v)\setminus\{u\}$. Since $z_u=x_u$, we deduce that $x\leq_v z$ and thus $f_v(x)\leq f_v(z)$. Since $f_v(x)=x_v=1$, we obtain $f_v(z)=1=x_v$, as required. If $x_v=0$ the proof is similar.  
\end{proof}

\begin{proof}[{\bf Proof of Theorem~\ref{uniqueness}}]
Let $x$ and $y$ be distinct fixed points of $f$, and let $I$ be the set of $v\in [n]$ such that $x_v\neq y_v$. By Lemma~\ref{alphabeta_lemma}, $G(x)[I]$ has a cycle $C$. It is thus sufficient to prove that $C$ has no special arc. Suppose, for a contradiction, that $C$ has a special arc $a=(uv,\e)$. By Lemma~\ref{canalizing}, 
\[
f_v(z)=x_v\text{ for all $z\in\B^n$ such that $z_u=x_u$}.
\]
Since $G(y)[I]=G(x)[I]$, by applying the same lemma, we get 
\[
f_v(z)=y_v\text{ for all $z\in\B^n$ such that $z_u=y_u$}.
\]
Since $x_v\neq y_v$ and $x_u\neq y_u$ we deduce that, for all $z\in\B^n$, $f_v(z)=x_v$ if and only if $z_u=x_u$. Hence, $f_v$ only depends on $x_u$, so $a$ is the unique in-coming arc of $v$ in $G$. Thus $v$ is a source of $G\setminus a$, and this contradicts the fact that $a$ is a special arc. Thus $C$ has no special arc. 
\end{proof}

Theorems~\ref{newrule1} and \ref{renewrule1} stated in the introduction are easy corollaries of Theorem~\ref{uniqueness}.

\begin{proof}[{\bf Proof of Theorems~\ref{newrule1} and  \ref{renewrule1}}] 
In Theorem~\ref{newrule1}, each positive cycle $C$ has an arc $a$ satisfying some conditions that trivially imply that $a$ is a special arc of $C$. In Theorem~\ref{renewrule1}, each positive cycle $C$ has a vertex $v$ satisfying some conditions that trivially imply that the arc $a$ of $C$ with terminal vertex $v$ is a special arc of $C$. Thus in both cases, every positive cycle has a special arc and, by Theorem~\ref{uniqueness}, $f$ has at most one fixed point. 
\end{proof}

\begin{remark}\label{rem:DR1}
{\em An easy corollary of Theorem~\ref{uniqueness} is the following: 
\begin{quote}
{\em If $G$ is strong, has a unique positive cycle $C$ and at least one negative cycle, then $f$ has at most one fixed point} (since every arc $a=(uv,\e)$ of $C$ such that $v$ is of in-degree at least two is then a special arc of $C$).
\end{quote}
This has been proved in \cite{DR12} under the additional assumptions that there exists a vertex $v$ meeting every cycle and that all the vertices $u\neq v$ have in-degree one.}
\end{remark}

\begin{remark}
{\em Aracena proved in \cite{A08} the following: 
\begin{quote}
{\em If $G$ has a non-trivial initial component without negative cycle, then $f$ has no fixed~point.} 
\end{quote}
This follows directly from Lemma~\ref{source_lemma}. Indeed, if $G$ has a non-trivial component $I$ and $f(x)=x$ then, by Lemma~\ref{source_lemma}, $G(x)[I]$ has no source and thus at least one cycle $C$, which is a positive cycle of $G[I]$. Furthermore, if $f$ canalizes no arc of $C$, then we deduce from Lemma~\ref{canalizing} that $C$ has no special arc. Hence, we get a weaker sufficient condition for the absence of fixed point (which however does not only depend on $G$):}
\begin{quote}
If $G$ has a non-trivial component $I$ in which every positive cycle has a special arc, and if $f$ canalizes no arc that belong to a positive cycle of $G[I]$, then $f$ has no fixed~point.
\end{quote}
\end{remark}

\begin{remark}
{\em The notion of special arc relies on connections between positive and negative cycles in the following sense:
\begin{quote}
{\em If $C$ has a special arc, then either $C$ intersects a negative cycle or $G$ has a non-trivial initial component with only negative cycles} (and thus $f$ has no fixed point by the theorem of Aracena mentioned just above). 
\end{quote}
Indeed, suppose that $C$ has a special arc $a=(uv,\e)$. Since $v$ is not a source of $G\setminus a$, $G$ has an arc $a'=(u'v,\e')\neq a$. If $u'$ and $C$ are not in the same  component, then $G$ has a path $P$ from an initial component $I$ to $v$ without vertex in $C\setminus v$. Since $a$ is a special arc of $C$, we deduce that $I$ is not a source and has no positive cycle. So $I$ is a non-trivial initial component of $G$ with only negative cycles. If $u'$ and $C$ are in the same component, then $G$ has a path $P$ from $v$ to $u'$, and thus $C'=P\cup a'$ is a cycle of $G\setminus a$ containing $v$. Thus $C'$ is negative (since $a$ is special) and $C'$ intersects $C$.}
\end{remark}

%%%%%%%%%%%%%%%%%%%%%%%%%%%%%%%%%%%%%%%%%%%%%%%%%%%%%%
\section{Upper-bounds}\label{sec:upperbound}
%%%%%%%%%%%%%%%%%%%%%%%%%%%%%%%%%%%%%%%%%%%%%%%%%%%%%%

Let $\tilde\tau^+$ be the minimum size of a set of vertices $I\subseteq [n]$ such that, in $G^I$, every positive cycle has a special arc. Let $\tilde g^+$ be the minimum length of a positive cycle of $G$ without special arc (with the convention that $\tilde g^+=\infty$ if such a cycle does not exist). Below, we prove that $2^{\tilde\tau^+}$ and $A(n,\tilde g^+)$ are upper bounds one the number of fixed points of $f$. Since we always have $\tilde\tau^+\leq \tau^+$ and $g^+\leq \tilde g^+$, these upper-bounds improve the upper bounds $2^{\tau^+}$ and $A(n,\tilde g^+)$ mentioned in the introduction. Actually, the gap can be arbitrarily large since if $G$ is as in Example \ref{ex1} then $\tilde\tau^+=0$ and $\tilde g^+=\infty$, so that $2^{\tilde \tau^+}=A(n,\tilde g^+)=1$, while both $2^{\tau^+}$ and $A(n,g^+)$ are exponential with $n$, as explained in Remark~\ref{rem:upperbound}. 

%%%%%%%%%%%%%%%%%%
\begin{corollary}\label{newupperbound}
$f$ has at most $\min(2^{\tilde \tau^+},A(n,\tilde g^+))$ fixed points. 
\end{corollary} 
%%%%%%%%%%%%%%%%%%

\begin{proof}[{\bf Proof}]
If $x$ and $y$ are two distinct fixed points of $f$ then $C$ has a positive cycle without special arc such that $x_v\neq y_v$ for every vertex $v$ of $C$ (Theorem \ref{uniqueness}). Thus $\tilde g^+$ is at most the length of $C$, which is at most the Hamming distance between $x$ and $y$. Thus $f$ has indeed at most $A(n,\tilde g^+)$ fixed points. 

\medskip
Let us now prove that $f$ has at most $2^{\tilde\tau^+}$ fixed points. Let $I$ be a set of vertices of size $\tilde\tau^+$ such that, in $G^I$, every positive cycle has a special arc. Let $X$ be the set of fixed points of $f$ and suppose, for a contradiction, that $|X|>2^{|I|}$. Then the function from $X$ to $2^I$ that maps $x$ on $x_I$ is not an injection, thus there exists distinct $x,y\in X$ such that $x_I=y_I$. Let $f'$ be the $n$-component network defined by $f'_v=f_v$ for every $v\in [n]\setminus  I$ and $f'_v=\cst=x_v$ for every $v\in I$. Then $x$ and $y$ are fixed points of $f'$ and thus, by Theorem  \ref{uniqueness}, the interaction graph of $f'$, which is precisely $G^I$, has a positive cycle without special arc, a contradiction. Thus $|X|\leq 2^{|I|}$ as required. 
\end{proof}

%%%%%%%%%%%%%%%%%%%%%%%%%%%%%%%%%%%%%%%%%%%%%%%%%%%%%%%%%%%%%%%%%%%%
\section{Existence results}\label{sec:existence}
%%%%%%%%%%%%%%%%%%%%%%%%%%%%%%%%%%%%%%%%%%%%%%%%%%%%%%%%%%%%%%%%%%%%

We define a {\bf two-coloring} of $G$ as a point $x\in\B^n$ such that $G(x)=G$. In other words, regarding $x_v$ as the color of vertex $v$, $x$ is a two-coloring if all the negative arcs link vertices with the distinct colors, and all the positive arcs link vertices with the same color. If $G$ has only negative arcs, we then recover the usual notion of (proper) two-coloring for unsigned graphs. Obviously, $x$ is a two-coloring if and only if $\overline{x}$ is a two-coloring. 

\medskip
We denote by $G^*$ the signed digraph obtained from $G$ by adding an arc $(uv,\e)$ for every arc $(vu,\e)$ of $G$. Hence, $G^*$ can be regarded as the symmetric (or undirected) version of $G$. A well known theorem of Cartwright and Harary \cite{CH56} asserts the following: 

\begin{theorem}\label{thm_Harary}
$G^*$ has no negative cycle if and only if $G$ has a two-coloring.
\end{theorem}

We will also use the following easy observation on the negative cycles of $G^*$ and $G$. 

\begin{lemma}\label{lem_GG*}
If $G$ is strong, then $G$ has a negative cycle if and only if $G^*$ has a negative cycle.
\end{lemma}

Below are two easy applications of the fact that $f_v$ is non-decreasing with $\leq_v$.

%%%%%%%%%%%%%%%%%%
\begin{lemma}\label{lem_ex_1}
Let $v$ be a vertex of $G$ of in-degree at least one and $x\in\B^n$. If all the in-coming arcs of $v$ are in $G(x)$ then $f_v(x)=x_v$. 
\end{lemma}
%%%%%%%%%%%%%%%%%%

\begin{proof}[{\bf Proof}]
Suppose that all the in-coming arcs of $v$ are in $G(x)$. If $x_v=0$, then $x_u=0$ for all $u\in\In^+(v)$ and $x_u=1$ for all $u\in\In^-(v)$. Hence, for all $y\in\B^n$, we have $x\leq_v y$ and thus $f_v(x)\leq f_v(y)$. We deduce that if $f_v(x)=1$, then $f_v(y)=1$ for all $y\in\B^n$. But then $f_v$ is a constant and this is equivalent to say that $v$ is a source of $G$, a contradiction. Therefore $f_v(x)=0=x_v$ as required. If $x_v=1$ the proof is similar. 
\end{proof}

%%%%%%%%%%%%%%%%%%
\begin{lemma}\label{lem_ex_2}
Let $a=(uv,\e)$ be an arc of $G$, and let $x\in\B^n$ such that $f_v(x)\neq x_v$. Suppose that all the in-coming arcs of $v$ distinct from $a$ are in $G(x)$. Then $f_v(z)\neq x_v$ for all $z\in\B^n$ such that $z_u=x_u$. Furthermore, $x_u\neq x_v$ if and only if $\e=+$. 
\end{lemma}
%%%%%%%%%%%%%%%%%%

\begin{proof}[{\bf Proof}]
If $x_v=0$, then $x_w=0$ for all $w\in\In^+(v)\setminus \{u\}$ and $x_w=1$ for all $w\in\In^-(v)\setminus\{u\}$. Hence, for all $z\in\B^n$ such that $z_u=x_u$, we have $x\leq_v z$ and thus $x_v<f_v(x)\leq f_v(z)$. Therefore, $f_v(z)\neq x_v$ as required. If $x_v=0$ the proof is similar. Furthermore, $a$ is not in $G(x)$, since otherwise $f_v(x)=x_v$ by Lemma \ref{lem_ex_1}. Thus $x_u\neq x_v$ if $\e=+$ and $x_u=x_v$ otherwise. 
\end{proof}

We are now in position to prove Theorem~\ref{newrule2}, that we restate from the introduction. 

\setcounter{theorem}{4}

\begin{theorem}
If every negative cycle of $G$ has an arc $a=(uv,\epsilon)$ such that $G\setminus a$ has a non-trivial initial strong component containing $v$ and only positive cycles, then $f$ has at least one fixed point.
\end{theorem}

\begin{proof}[{\bf Proof}]
We proceed by induction on the number of negative cycles. If $G$ has no negative cycle, then $f$ has at least one fixed point by Theorem \ref{rule2}. So suppose that $G$ has at least one negative cycle $C$ and satisfies the condition of the theorem. By hypothesis, $C$ has an arc $a=(uv,\epsilon)$ such that $G\setminus a$ has a non-trivial initial component $I$ containing $v$ and only positive cycles. Suppose that $I$ is maximal in the following sense: for every arc $a'=(u'v',\epsilon')$ that belongs to a negative cycle and such that $G\setminus a'$ has a non-trivial initial component $I'$ containing $v'$ and only positive cycles, $I$ is not a strict subset of $I'$. 

\medskip
Let us prove that \[
\text{\em $G^I$ satisfies the condition of the theorem.}
\]
Let $C'$ be a negative cycle of $G^I$ and let $a'=(u'v',\epsilon')$ be an arc of $C'$ such that $G\setminus a'$ has a non-trivial initial component $I'$ containing $v'$ and only positive cycles. It is sufficient to prove that $I'$ is an initial component of $G^I\setminus a'$. Since $G[I]$ is strong and $v'\not\in I$ (because $C'$ and $I$ are disjoint), $G[I]\setminus a'$ is strong. Thus if $I\cap I'\neq\emptyset$, then $G[I\cup I']\setminus a'$ is strong, and since $I'$ is an initial component of $G\setminus a'$ we deduce that $I\subseteq I'$. Since $v'\in I'\setminus I$, this contradicts the maximality of $I$. Thus $I\cap I'=\emptyset$ and it is then straightforward to show that $I'$ is an initial component of $G^I\setminus a'$. So $G^I$ satisfies the condition of the theorem, as required. 

\medskip
Since $G[I]\setminus a$ is strong and has no negative cycle, it follows that $(G[I]\setminus a)^*$ has no negative cycle (by Lemma~\ref{lem_GG*}), and thus it has a two-coloring $\chi\in \B^I$ (by Theorem~\ref{thm_Harary}). Consider the $n$-component networks $\tilde f$ and $\hat f$ defined as follows: 
\[
\begin{array}{c}
\left\{
\begin{array}{ll}
\tilde f_w=\cst=1&\text{for all }w\in I\text{ with }\chi_w=1\\
\tilde f_w=\cst=0&\text{for all }w\in I\text{ with }\chi_w=0\\
\tilde f_w=f_w&\text{for all }w\notin I\\
\end{array}
\right.
\\[8mm]
\left\{
\begin{array}{ll}
\hat f_w=\cst=1&\text{for all }w\in I\text{ with }\chi_w=0\\
\hat f_w=\cst=0&\text{for all }w\in I\text{ with }\chi_w=1\\
\hat f_w=f_w&\text{for all }w\notin I.\\
\end{array}
\right.
\end{array}
\]
Since $G^I$ is the interaction graph of $\tilde f$ and $\hat f$, and since $G^I$ satisfies the condition of the theorem, by induction hypothesis $\tilde f$ has a fixed point $x$ and $\hat f$ has a fixed point $y$. Obviously we have $x_I=\chi$ and $y_I=\overline{\chi}$.  

\medskip
Furthermore, 
\begin{equation}\label{eq:ab}
\forall w\neq v,\qquad f_w(x)=x_w\quad\text{and}\quad f_w(y)=y_w.
\end{equation}
Indeed, if $w\notin I$ then $f_w(x)=\tilde f_w(x)=x_w$. Suppose now that $w\in I$. Since $I$ is a non-trivial initial component of $G\setminus a$ and $w\neq v$, all the in-coming arcs of $w$ in $G$ are in $G[I]\setminus a$. Since $x_I=\chi$ is a two-coloring of $G[I]\setminus a$, we have $G[I]\setminus a=G[I](x)\setminus a$, and we deduce that all the in-coming arcs of $w$ are in $G(x)$. Hence, by Lemma~\ref{lem_ex_1}, $f_w(x)=x_w$. We prove with similar arguments that $f_w(y)=y_w$ for all $w\neq v$, using the fact that $y_I=\overline{\chi}$ is also a two-coloring of $G[I]\setminus a$. This proves (\ref{eq:ab}). 

\medskip
Let us now prove that either $x$ or $y$ is a fixed point of $f$. Suppose, for a contradiction, that $f(x)\neq x$ and $f(y)\neq y$. Then, according to (\ref{eq:ab}), we have $f_v(x)\neq x_v$, and since $I$ is an initial component of $G\setminus a$, all the in-coming arcs of $v$ distinct from $a$ are in $G[I]\setminus a$. Since $x_I=\chi$ we have $G[I](x)\setminus a=G[I]\setminus a$ and we deduce that all the in-coming arcs of $v$ distinct from $a$ are in $G(x)$. Thus, according to Lemma~\ref{lem_ex_2}, we have 
\[
f_v(z)\neq x_v\text{ for all $z\in\B^n$ such that $z_u=x_u$}. 
\]
We prove with similar arguments that 
\[
f_v(z)\neq y_v\text{ for all $z\in\B^n$ such that $z_u=y_u$}. 
\]
Since $x_v\neq y_v$ we have $x_u\neq y_u$ and we deduce that, for all $z\in\B^n$, $f_v(z)=x_v$ if and only if $z_u\neq x_u$. Thus $f_v(x)$ only depends on $x_u$, and thus $a$ is the only arc of $G$ with terminal vertex~$v$. But then $I$ cannot be a non-trivial component of $G\setminus a$, a contradiction. Thus $x$ or $y$ is a fixed point of $f$. 
\end{proof}

\begin{remark}\label{rem:DR2}
{\em An easy corollary of Theorem~\ref{newrule2} is the following: 
\begin{quote}
{\em If $G$ is strong, has a unique negative cycle $C$ and at least one positive cycle, then $f$ has at least one fixed point} (since every arc $a=(uv,\e)$ of $C$ such that $v$ is of in-degree at least two satisfies the condition of Theorem~\ref{newrule2}).
\end{quote}
This has been proved in \cite{DR12} under the additional assumptions that there exists a vertex $v$ meeting every cycle, and that all the vertices $u\neq v$ have in-degree one.}
\end{remark}

\begin{remark}
{\em Aracena proved in \cite{A08} the following theorem: 
\begin{quote}
{\em If $G$ is strong and has no negative cycle, then $f$ has two fixed points with Hamming distance $n$.
} 
\end{quote}
This follows directly from Harary's theorem and Lemma~\ref{lem_ex_1}. Indeed, if $G$ is strong and has no negative cycle, then $G$ has a two-coloring $x$ (by Theorem~\ref{thm_Harary} and Lemma \ref{lem_GG*}), and if $G$ has no source, then $f(x)=x$ by Lemma~\ref{lem_ex_1}; and since $\overline{x}$ is also a two-coloring, we also get $f(\overline{x})=\overline{x}$. From this result, we easily deduce a statement with a weaker condition and a weaker conclusion: 
\begin{quote}
{\em If $G$ has no negative cycle and a non-trivial initial component, then $f$ has at least two fixed points}. 
\end{quote}
Now, in the proof of Theorem~\ref{newrule2}, we show that if $f(x)\neq x$ or $f(y)\neq y$ then the arc $a$ is canalized by~$f$, and thus we get the following generalization (which however does not only depend on $G$):} 
\begin{quote}
If $G$ satisfies the condition of Theorem~\ref{newrule2} and has a non-trivial initial component $I$ such that $f$ canalizes no arc that belongs to a negative cycle of $G[I]$, then $f$ has at least two fixed points.
\end{quote}
\end{remark}

\begin{remark}
{\em A {\em kernel} in a digraph $D=(V,A)$ is an independent set of vertices $K\subseteq V$ such that, for every $v\in V\setminus K$, $D$ has an arc from $v$ to $K$. Not every digraph has a kernel, and the following well known theorem of Richardson \cite{R53} asserts that 
\begin{quote}
{\em If $D$ has no odd cycle, then $D$ has at least one kernel}. 
\end{quote}
Here, the parity of a cycle is the parity of its length. Many generalizations of this results have been established, see for instance \cite{GSNL84,GSLS15} and the references therein. In \cite{RR13} a correspondence with Boolean networks shows that Richardson's theorem is a corollary of Theorem~\ref{rule2}, and using this correspondence, we deduce from Theorem~\ref{newrule2} a new generalization of Richardson's theorem:}
\begin{quote}
If every odd cycle of $D$ has an arc $uv$ such that $D\setminus uv$ has a non-trivial initial component containing $v$ and only even cycles, then $D$ has at least one kernel.
\end{quote}
\end{remark}

We now go to the proof of Theorem~\ref{strong}, already stated in the introduction, that gives a new sufficient condition for the existence of two fixed points with Hamming distance $n$, using again information on the arcs canalized by $f$. 

\begin{theorem}
If $G$ is strongly connected, has a unique negative cycle and at least one positive cycle, and if $f$ canalizes no arc that belongs to the negative cycle, then $f$ has two fixed points with Hamming distance $n$. 
\end{theorem}

We need the following lemma and few definitions. A {\em closed walk} $W$ in $G$ is a sequence of paths of $G$, say $W=(P_1,\dots,P_k)$, such that, for all $1\leq i<k$, the last vertex of $P_i$ is the first vertex of $P_{i+1}$, and such that the last vertex of $P_k$ is the first vertex of $P_1$. The sign of a walk is the product of the signs of its paths and is denoted $s(W)\in\{-1,1\}$. It is easy to see that {\em if $W$ is negative, then $P_1\cup\cdots\cup P_k$ has a negative cycle.} 

%%%%%%%%%%%%%%%%%%
\begin{lemma}\label{graph}
If $G$ has a unique negative cycle, then it has an arc that belongs to no positive~cycle. 
\end{lemma} 
%%%%%%%%%%%%%%%%%%

\begin{proof}[{\bf Proof}]
Let $C$ be the unique negative cycle of $G$ and suppose, for a contradiction, that every arc of $C$ belongs to a positive cycle. A path $P$ of $G$ from $u$ to $v$ is said to be an {\em alternative path} if $u$ and $v$ are distinct vertices of $C$ and if $P$ and $C[u,v]$ are arc-disjoint. 

\medskip
We first prove that {\em for every arc $a=(uv,\e)$ of $C$, there exists an alternative path $P$ from $u'$ to $v'$ such that $a$ is an arc of $C[v',u']$}. Indeed, by hypothesis, $a$ is contained in at least one positive cycle $F$. Let $a'$ be the first arc of $F[v,u]$ that is not in $C$, let $u'$ be the initial vertex of $a'$, and let $v'$ be the first vertex of $F[u',u]$ distinct from $u'$ that belongs to $C$. Then $F[u',v']$ is an alternative path, and since $C[u',v']\subseteq C[v,u]$, $a$ is an arc of $C[v',u']$. 

\medskip
Furthermore, 
\begin{equation}\label{eq:sign}
{\text{\em Each alternative path $P$ from $u$ to $v$ has the same sign as $C[v,u]$.}}
\end{equation}
Indeed, if $P$ and $C[v,u]$ have not the same sign, then $(P,C[v,u])$ is a negative closed walk of $H=P\cup C[v,u]$. Thus $H$ has a negative cycle $F$, and since $P$ and $C[u,v]$ are arc-disjoint, $F\neq C$, a contradiction. 

\medskip

Let $P_1$ be an alternative path from $u_1$ to $v_1$ that maximizes the length of $C[v_1,u_1]$. Let $a$ the arc of $C$ with terminal vertex $v_1$, and let $P_2$ be an alternative path from $u_2$ to $v_2$ such that $a$ is an arc of $C[v_2,u_2]$ (we have proved that such an alternative path exists).  Let us prove that 
\[
C[u_2,v_2]\subseteq C[v_1,u_1].
\]
Indeed, $u_2$ is a vertex of $C[v_1,u_1]$ since otherwise $C[v_1,u_1]$ is a strict subgraph of $C[v_2,u_2]$, a contradiction with our assumption on~$P_1$. Hence, since $v_2$ is a vertex of $C[u_2,v_1]\setminus\{v_1\}$, $v_2$ is either a vertex of $C[u_2,u_1]$ or a vertex of $C[u_1,v_1]\setminus \{u_1,v_1\}$. Suppose that $v_2$ is a vertex of $C[u_1,v_1]\setminus \{u_1,v_1\}$, and let $H=P_1\cup C[v_1,u_2]\cup P_2$. Obviously, $H$ contains a path $P$ from $u_1$ to $v_2$. Since $C[u_1,v_2]\subseteq C[u_1,v_1]$ and $C[u_1,v_2]\subseteq C[u_2,v_2]$, $P_1$ and $P_2$ are arcs-disjoint from $C[u_1,v_2]$, and since $C[v_1,u_2]\subseteq C[v_2,u_1]$, we deduce that $H$ is arc-disjoint from $C[u_1,v_2]$. Thus $P$ is an alternative path from $u_1$ to $v_2$, and since $C[v_1,u_1]$ is a strict subgraph of $C[v_2,u_1]$, this contradicts our assumption on $P_1$. Consequently, $v_2$ is a vertex of $C[u_2,u_1]$, and thus $C[u_2,v_2]\subseteq C[v_1,u_1]$.

\medskip
We are now in position to prove the lemma. Consider the subgraph
\[
H=P_1\cup C[v_1,u_2]\cup P_2\cup C[v_2,u_1]. 
\]
It contains a closed walk $W$ of sign 
\[
s(W) =s(P_1)s(C[v_1,u_2])s(P_2)s(C[v_2,u_1])
\]
and by (\ref{eq:sign}) we obtain
\[
s(W)=s(C[v_1,u_1])s(C[v_1,u_2])s(C[v_2,u_2])s(C[v_2,u_1]).
\]
Since $C[u_2,v_2]\subseteq C[v_1,u_1]$ we have 
\[
s(C[v_1,u_1])=s(C[v_1,u_2])s(C[u_2,v_2])s(C[v_2,u_1])
\]
thus
\[
s(W) =s(C[u_2,v_2])s(C[v_2,u_2])s(C[v_1,u_2])^2s(C[v_2,u_1])^2
=s(C[u_2,v_2])s(C[v_2,u_2])=s(C).
\]
Hence, $W$ is a negative closed walk, thus $H$ has an negative cycle. We deduce that $C\subseteq H$. Since $P_2$ is an alternative path, $C[u_2,v_2]$ is arc-disjoint from $P_2$, and since $C\subseteq H$ we deduce that $C[u_2,v_2]\subseteq P_1$. Let $v'_1$ be the first vertex of $P_1$ that belongs to $C[v_1,u_2]$, and let $P'_1$ be the path from $u_1$ to $v'_1$ contained in $P_1$. Clearly, $P'_1$ is an alternative path arc-disjoint from $C[u_2,v_2]$. Now, 
\[
H'= P_1'\cup C[v'_1,u_2]\cup P_2\cup C[v_2,u_1]
\]
is arc-disjointed from $C[u_2,v_2]$, and it contains a closed walk $W'$ of sign 
\[
s(W') =s(P'_1)s(C[v'_1,u_2])s(P_2)s(C[v_2,u_1]).
\]
From (\ref{eq:sign}) we get 
\[
s(W') =s(C[v'_1,u_1])s(C[v'_1,u_2])s(C[v_2,u_2])s(C[v_2,u_1]).
\]
Since $C[u_2,v_2]\subseteq C[v'_1,u_1]$ we have 
\[
s(C[v'_1,u_1])=s(C[v'_1,u_2])s(C[u_2,v_2])s(C[v_2,u_1])
\]
thus
\[
s(W') =s(C[u_2,v_2])s(C[v_2,u_2])s(C[v'_1,u_2])^2s(C[v_2,u_1])^2=s(C).
\]
Hence, $W'$ is an negative closed walk, thus $H'$ has a negative cycle $F$. Since $H'$ is arc-disjoint from $C[u_2,v_2]$, $F\neq C$, a contradiction.
\end{proof}

\begin{proof}[{\bf Proof of Theorem \ref{strong}}]
Let $C$ be the unique negative cycle of $G$. Let $A$ be the set of arcs of $C$ that belong to no positive cycle, which is not empty by Lemma~\ref{graph}. Suppose that $a=(uv,\e)\in A$, and let $a'$ be the arc succeeding $a$ in $C$. If $v$ is of in-degree one, then every cycle containing $a'$ contains also $a$ and thus $a'\in A$. Since some vertex of $C$ is of in-degree at least two, we deduce that there exists an arc in $A$, say $a=(uv,\e)$,  such that $v$ is of in-degree at least two. 

\medskip
Let $I_1,\dots,I_k$ be an enumeration of the components of $G\setminus a$ in the topological order (that is, in such a way that for all $1\leq i<j\leq k$, $G\setminus a$ has no arc from $I_j$ to $I_i$). Since $G$ is strong, $I_1$ (resp. $I_k$) is the unique initial (resp. terminal) component of $G\setminus a$, and $v\in I_1$ (resp. $u\in I_k$). Suppose that there exists two distinct arcs leaving a component $I_p$, $1\leq p<k$. Then $G\setminus a$ has two distinct paths from $v$ to $u$, and this contradict that fact that $C$ is the unique cycle containing~$a$. We deduce that, for each $1\leq p<k$, there is a unique arc, say $a_p=(u_pv_p,\e_p)$, leaving the component $I_p$. Then $v_p\in I_{p+1}$, since otherwise there is an initial component distinct from $I_1$. Since every component $I_p$ has no negative cycle, we deduce that $(G\setminus a)^*$ has no negative cycle, and so $G\setminus a$ has a two-coloring $x\in\B^n$. 

\medskip
Since $v$ is of in-degree at least two in $G$, $G\setminus a$ has no source, and thus we deduce from Lemma~\ref{lem_ex_1} that $f_w(x)=x_w$ for all $w\neq v$. From Lemma~\ref{lem_ex_2}, if $f_v(x)\neq x_v$ then $f$ canalizes $a$, a contradiction. Thus $f(x)=x$ and we prove with similar arguments that $f(\overline{x})=\overline{x}$. 
\end{proof}

%%%%%%%%%%%%%%%%%%%%%%%%%%%%%%%%%%%%%%%%%%%%%%%%%%%%%%%%%%%%%%%%%%%%
\section{Concluding remarks}\label{sec:conclusion}
%%%%%%%%%%%%%%%%%%%%%%%%%%%%%%%%%%%%%%%%%%%%%%%%%%%%%%%%%%%%%%%%%%%%

We have established new sufficient conditions, expressed on $G$ only, for the uniqueness (resp. existence) of a fixed point of $f$. In order to know if it is reasonable to think about a characterization, it could be interesting to study the complexity of the following decision problem: {\em Given an signed digraph $G$, is it true that all the Boolean networks with $G$ as interaction graph have at most (resp. at least) one fixed point?}

\medskip
We have established an upper-bound on the number of fixed points, namely $2^{\tilde\tau^+}$, that improve the classic bound $2^{\tau^+}$. This new bound raise some questions. Let $\max(G)$ be the maximal number of fixed points in a Boolean network with $G$ as interaction graph. The signed digraph in Example \ref{ex1} shows that $\max(G)$ does not necessarily increase with $\tau^+$. We can ask if the situation is identical with $\tilde\tau^+$: {\em Does $\max(G)$ necessarily increase with $\tilde\tau^+$? In other words, is there exists an unbounded function $h$, independent of $G$, such that $h(\tilde\tau^+)\leq \max(G)$?} 

\medskip
Let $\Gamma(f)$ be the digraph with vertex set $\B^n$ and with an arc from $x$ to $y$ if there exists $v\in[n]$ such that $y=\overline{x}^v$ and $f_v(x)\neq x_v$. This digraph $\Gamma(f)$ is usually called the {\em asynchronous state graph} of $f$, and is a classical model for the dynamic of gene network \cite{TA90,TK01}. The terminal components of $\Gamma(f)$ are then regarded as the attractors of the system and called {\em attractors} of $\Gamma(f)$. Hence, fixed points are attractors of cardinality one. Attractors of size at least two are called {\em cyclic attractors}. It is known that the upper-bounds $2^{\tau^+(G)}$ and $A(n,g^+(G))$ are actually upper-bounds one the number of attractors in $\Gamma(f)$ \cite{R09}. {\em Is it also the case with the new bounds $2^{\tilde\tau^+(G)}$ and $A(n,\tilde g^+(G))$?} Besides, it is known that if $G$ has no negative cycle, then $\Gamma(f)$ has no cyclic attractor (and this trivially implies that $f$ has a fixed point) \cite{R10}. {\em Is the weaker condition in Theorem~\ref{newrule2} also sufficient for the absence of cyclic attractors in $\Gamma(f)$?} 
%We can also think about generalization of the results to the non-Boolean discrete case. 

\medskip
In the proof of Theorem~\ref{strong}, we show that a signed digraph with a unique negative cycle is very-well structured and, in particular, is not $2$-arc-strongly connected. Continuing this direction, it could be interesting to study the structure of signed digraphs in which all the negative cycles are vertex disjoint, in order to obtain a sufficient condition for the presence of a fixed point with the spirit of the condition in Theorem~\ref{renewrule1}. 

\paragraph{Acknowledgment} I wish to thank Emmanuelle Seguin and Isabelle Soubeyran for stimulating discussions. 

\bibliographystyle{plain}
\bibliography{bib}

\end{document}